\documentclass[a4paper,10pt]{amsart}
\usepackage{graphicx}
\usepackage{amssymb,amsmath,amstext,amscd}
\theoremstyle{plain}
\newtheorem{thm}{Theorem}[section]

\newtheorem*{main theorem}{Theorem}

\theoremstyle{definition}

\begin{document}

\title{A characterization of causal automorphisms on two-dimensional Minkowski spacetime}
\author{Do-Hyung Kim}
\address{Department of Mathematics, College of Natural Science, Dankook University,
San 29, Anseo-dong, Dongnam-gu, Cheonan-si, Chungnam, 330-714,
Republic of Korea} \email{mathph@dankook.ac.kr}


\begin{abstract}
It is shown that causal automorphisms on two-dimensional Minkowski
spacetime can be characterized by the invariance of the wave
equations.
\end{abstract}

\maketitle

\section{Introduction} \label{section:1}

 In 1964, Zeeman clarified general forms of causal automorphisms
 on $\mathbb{R}^n_1$ with $n \geq 3$.(Ref. \cite{Zeeman}) As the
 title of his paper says, the result is that causality implies the
 Lorentz group.

 However, this is not the case for $\mathbb{R}^2_1$. In 2010,
 general forms of causal automorphisms on $\mathbb{R}^2_1$ is
 given and the result gives us much more symmetry than Zeeman's
 result.(Ref. \cite{CQG3}, \cite{CQG4}, \cite{Low} and
 \cite{Minguzzi})

 From this result, we can see that each component of the causal
 automorphisms satisfies the wave equation and so a new question is
 raised by Low.(Ref. \cite{Low}) His question is essentially how
 to characterize causal automorphisms by wave equations. A
 characterization of causal automorphisms on $\mathbb{R}^n_1$ with
 $n \geq 3$ is given in Ref. \cite{Wave}. The result is that
 causal automorphisms on
 $\mathbb{R}^n_1$ with $n \geq3$ can be characterized by the invariance of the wave
 equations. (Ref. \cite{Wave}). Physically, the invariance of the wave
 equations means the principle of the constancy of the speed of
 the light, which is the second postulate of special relativity.
 Therefore, we can say that Einstein's second postulate implies
 the preservation of causal relations as well as the Lorentz
 group,
 when the dimensions of spacetimes are greater than or equal to 3.

 In this paper, it is shown that $C^2$-causal automorphisms on
 $\mathbb{R}^2_1$ can be obtained by the invariance of the wave
 equations. Therefore, in conclusion, we can say that causal
 automorphisms on Minkowski spacetime can be characterized as the
 invariance of the wave equations, regardless of the spacetime dimension.

\section{Preliminaries} \label{section:2}

In Ref. \cite{CQG4}, the following is shown.

\begin{thm} \label{original}
Let $F : \mathbb{R}^2_1 \rightarrow \mathbb{R}^2_1$ be a causal
automorphism. Then, there exists unique homeomorphisms $\varphi$
and $\psi$ of $\mathbb{R}$, which are both increasing or both
decreasing such that if $\varphi$ and $\psi$ are increasing, then
we have $F(x,t)=\frac{1}{2}\big( \varphi(x+t) + \psi(x-t),
\varphi(x+t) - \psi(x-t) \big)$, of if $\varphi$ and $\psi$ are
decreasing, then we have $F(x,t)=\frac{1}{2} \big(
\varphi(x-t)+\psi(x+t), \varphi(x-t)-\psi(x+t) \big)$.

Conversely, for any given homeomorphisms $\varphi$ and $\psi$ of
$\mathbb{R}$, which are both increasing or both decreasing, the
function $F$ defined as above is a causal automorphism on
$\mathbb{R}^2_1$.
\end{thm}
\begin{proof}
See Theorem 2.2 in Ref. \cite{CQG4}.
\end{proof}

In Ref. \cite{Low}, Low has shown that if we use null coordinates,
the above result can be simplified.

If we use null coordinate system as $u=x+t$, $v=x-t$, then the
above Theorem can be expressed as following.

\begin{thm} \label{modified}
Let $(U,V) = F(u,v)$ be a causal automorphism on $\mathbb{R}^2_1$.
Then there exist unique homeomorphisms $\varphi$ and $\psi$ on
$\mathbb{R}$, which are both increasing, or both decreasing such
that, if $\varphi$ and $\psi$ are increasing, then we have $F(u,v)
= \big( \varphi(u), \psi(v) \big)$, or if $\varphi$ and $\psi$ are
decreasing, then we have $F(u,v) = \big( \varphi(v), \psi(u)
\big)$.

Conversely, for any given homeomorphisms $\varphi$ and $\psi$ of
$\mathbb{R}$, which are both increasing or both decreasing, the
function $F$ defined as above is a causal automorphism on
$\mathbb{R}^2_1$.
\end{thm}

\section{Invariance of wave equations} \label{section:3}

If we use null coordinates $u$ and $v$, then the wave equation
$\frac{\partial^2 \theta}{\partial x^2} - \frac{\partial^2
\theta}{\partial t^2} = 0 $ can be simplified by $\frac{\partial^2
\theta}{\partial u
\partial v} = 0$.

We now prove the main Theorem.

\begin{thm}
Let $F : \mathbb{R}^2_1 \rightarrow \mathbb{R}^2_1$ be a $C^2$
diffeomorphism given by $(\sigma,\tau) = F(u,v)$ in terms of null
coordinates. Then the necessary and sufficient condition for $F$
to be a causal automorphism is that $\frac{\partial
(\sigma-\tau)}{\partial t} \geq 0$, $\sigma_u \tau_v + \sigma_v
\tau_u \geq 0$ and, for any $C^2$ function $\theta$ on
$\mathbb{R}^2_1$, $\theta_{uv}=0$ if and only if $\theta_{\sigma
\tau}=0$.
\end{thm}

\begin{proof}

For the proof of the necessary condition, we only need
straightforward calculations and Theorems in Section
\ref{section:2} and so we omit the proof.

To prove  sufficient part we note that, since $F$ is a
diffeomorphism, we can use $\sigma$ and $\tau$ as a new coordinate
system on $\mathbb{R}^2_1$ and its Jacobian $J(F)=\sigma_u \tau_v
- \sigma_v \tau_u$ is nowhere zero.

Step 1 :
 By the chain rule, we have the following relation.

\begin{eqnarray*}
 \frac{\partial^2 \theta}{\partial v \partial u} &=& \Big(
\frac{\partial^2 \theta}{\partial \sigma^2} \frac{\partial
\sigma}{\partial v} \frac{\partial \sigma}{\partial u} +
\frac{\partial^2 \theta}{\partial \tau \partial \sigma}
\frac{\partial \tau}{\partial v} \frac{\partial \sigma}{\partial
u} + \frac{\partial \theta}{\partial \sigma}\frac{\partial^2
\sigma}{\partial v \partial u} \Big) \\
&+& \Big( \frac{\partial^2 \theta}{\partial \sigma
\partial \tau}\frac{\partial
\sigma}{\partial v}\frac{\partial \tau}{\partial u} +
\frac{\partial^2 \theta}{\partial \tau^2}\frac{\partial
\tau}{\partial v}\frac{\partial \tau}{\partial u} + \frac{\partial
\theta}{\partial t}\frac{\partial^2 \tau}{\partial v \partial u}
\Big) \,\,\,\,\ \cdots \cdots (*)
\end{eqnarray*}

If we let $\theta = \sigma$, then it satisfies $\theta_{\sigma
\tau}=0$ and thus it must satisfy $\theta_{uv}=0$. If we put this
into the above equation, we have $\frac{\partial^2
\sigma}{\partial v
\partial u}=0$. Since the coordinate transformation $F$ is $C^2$,
we have $\sigma = \int f(u) du + g(v)$.

Likewise, if we put $\theta = \tau$ in the above equation, we
obtain $\frac{\partial^2 \tau}{\partial v \partial u} =0$ and so
we have $\tau = \int h(u) du + j(v)$.

If we put $\frac{\partial^2 \sigma}{\partial v \partial u}=0$ and
$\frac{\partial^2 \tau}{\partial v \partial u} =0$ into $(*)$, we
have the following.

\begin{eqnarray*}
 \frac{\partial^2 \theta}{\partial v \partial u} = \Big(
\frac{\partial^2 \theta}{\partial \sigma^2} \frac{\partial
\sigma}{\partial v} \frac{\partial \sigma}{\partial u} +
\frac{\partial^2 \theta}{\partial \tau \partial \sigma}
\frac{\partial \tau}{\partial v} \frac{\partial \sigma}{\partial
u} \Big)  + \Big( \frac{\partial^2 \theta}{\partial \sigma
\partial \tau}\frac{\partial
\sigma}{\partial v}\frac{\partial \tau}{\partial u} +
\frac{\partial^2 \theta}{\partial \tau^2}\frac{\partial
\tau}{\partial v}\frac{\partial \tau}{\partial u}  \Big)
\,\,\,\,\ \cdots \cdots (**)
\end{eqnarray*}
\

Step 2 : If we let $\theta = \sigma^2$, then it satisfies
$\theta_{\sigma \tau}=0$ and thus it must satisfy $\theta_{uv}=0$.
If we put this into $(**)$, we have $\frac{\partial
\sigma}{\partial v}\frac{\partial \sigma}{\partial u}=0$. Since
$\sigma = \int f(u) du + g(v)$, we have $f(u) \cdot g^\prime(v) =
0$.

Likewise, if we put $\theta = \tau^2$ into $(**)$, we have
$\frac{\partial \tau}{\partial v}\frac{\partial \tau}{\partial
u}=0$. Since $\tau = \int h(u) du + j(v)$, we have $h(u)
\cdot j^\prime(v)=0$.\\

Step 3 : We consider two separate cases.

 Case (i) : We assume
that there exists $u_0$ such that $f(u_0) \neq 0$. Then, since
$f(u) \cdot g^\prime(v) = 0$ for all $u$ and $v$, we must have
$g^\prime(v)=0$ for all $v$. Then, since the Jacobian
$J(F)=f(u)j^\prime(v) - h(u)g^\prime(v)$ is nowhere zero, $f(u)$
must be nowhere zero. Also, by the same reason, $j^\prime(v)$ is
nowhere zero and thus $j(v)$ is a homeomorphism. Since $h(u) \cdot
j^\prime(v) = 0$, $h(u)=0$ for all $u$. In conclusion, we have
$\sigma = \int f(u) du $ and $\tau = j(v)$, which are
homeomorphisms.

Case (ii) : We now assume that $f$ is identically zero. Then,
since $J(F)$ is nowhere zero, $g^\prime(v)$ is nowhere zero and
thus $g(v)$ is a homeomorphism. Since $J(F)$ is nowhere zero,
$h(u)$ must be nowhere zero and thus from $h(u)j^\prime(v)=0$, we
have $j^\prime(v)=0$ for all $v$. In conclusion, we have $\sigma =
g(v)$ and $\tau = \int h(u) du $, which are homeomorphisms.\\

Step 4 : We now consider the condition $\frac{\partial
(\sigma-\tau)}{\partial t} \geq 0$ which implies time-orientation
preservation. By chain rule, we have the following.

\begin{eqnarray*}
0 \leq \frac{\partial (\sigma-\tau)}{\partial t} &=&
\frac{\partial \sigma}{\partial t} - \frac{\partial
\tau}{\partial t}\\
 &=& \Big(\frac{\partial \sigma}{\partial
u}\frac{\partial u}{\partial t}+\frac{\partial \sigma}{\partial
v}\frac{\partial v}{\partial t} \Big) - \Big(\frac{\partial
\tau}{\partial u}\frac{\partial u}{\partial t}+\frac{\partial
\tau}{\partial v}\frac{\partial v}{\partial t}\Big)\\
&=& \Big(\frac{\partial \sigma}{\partial u}-\frac{\partial
\sigma}{\partial v} \Big) - \Big(\frac{\partial \tau}{\partial
u}-\frac{\partial \tau}{\partial v} \Big)\\
&=& \Big( f(u) - g^\prime(v) \Big) - \Big( h(u)-j^\prime(v)
\Big)\\
\end{eqnarray*}

Therefore, we have $f(u) + j^\prime(v) \geq h(u) + g^\prime(v)$.\\

Step 5: We consider two separate cases.

Case (i) : We assume that there exists $u_0$ such that $f(u_0)
\neq 0$. Then, from Step 3, we know that $g^\prime$ and $h(u)$ are
identically zero and so from Step 4, we have $f(u)+j^\prime(v)
\geq 0$. Since $\sigma_u \tau_v + \sigma_v \tau_u \geq 0$, we have
$f(u) \cdot j^\prime(v) \geq 0$. Therefore, both $f$ and
$j^\prime$ are positive functions and so $\sigma = \int f(u) du $
and $\tau = j(v)$ are increasing homeomorphisms.

Case (ii) : We now assume that $f$ is identically zero. Then, from
Step 3, we know that $f$ and $j^\prime$ are identically zero and
so from step 4, we have $h(u)+g^\prime(v) \leq 0$. Since $\sigma_u
\tau_v + \sigma_v \tau_u \geq 0$, we have $h(u) \cdot g^\prime(v)
\geq 0$. Therefore, both $h$ and $g^\prime$ are negative functions
and so $\sigma = g(v)$ and $\tau = \int h(u) du$ are decreasing
homeomorphisms. \\

Step 6 : From Theorem \ref{modified}, $(\sigma, \tau) = F(u,v)$ is
a causal automorphism on $\mathbb{R}^2_1$.\\

\end{proof}

\section{Acknowledgement}

The present research was conducted by the research fund of Dankook
University in 2012.


\begin{thebibliography}{999}

%
\bibitem{Zeeman} E. C. Zeeman,
{\it Causality implies the Lorentz group}, J. Math. Phys. {\bf 5},
(1964) 490.




%
\bibitem{CQG3} D.-H. Kim,
{\it Causal automorphisms on two-dimensional Minkowski spacetime},
Class. Quantum. Grav. {\bf 27}, (2010) 075006.

%
\bibitem{CQG4} D.-H. Kim,
{\it The group of causal automorphisms}, Class. Quantum. Grav.
{\bf 27}, (2010) 155005.



%
\bibitem{Low} R. J. Low,
{\it Characterizing the causal automorphisms of 2D Minkowski
space}, Class. Quantum. Grav. {\bf 28}, (2011) 225009.

%
\bibitem{Minguzzi} A. Garc¢¥©¥a-Parrado and E. Minguzzi,
{\it Product posets and causal automorphisms of the plane}, Class.
Quantum. Grav. {\bf 28}, (2011) 147001.


%
\bibitem{Wave} D.-H. Kim,
{\it A characterization of causal automorphisms by wave
equations}, J. Math. Phys. {\bf 53}, (2012) 032507.









\end{thebibliography}
\end{document}